\documentclass[a4paper,11pt]{article}
\usepackage[T1]{fontenc}
\usepackage{lmodern}
\usepackage{amsmath,amssymb,amsthm}
\usepackage{pst-plot}
\usepackage{epsfig,pspicture,pstricks}
\usepackage{setspace}
\usepackage{float}
\usepackage{rotating}
\usepackage{geometry}
\newcommand*\samethanks[1][\value{footnote}]{\footnotemark[#1]}

\newtheorem{theorem}{Theorem}
\newtheorem{definition}[theorem]{Definition}
\newtheorem{lemma}[theorem]{Lemma}
\newtheorem{corollary}[theorem]{Corollary}
\newtheorem{proposition}[theorem]{Proposition}
\newtheorem{remark}[theorem]{Remark}
\newtheorem{example}[theorem]{Example}

\title{Directed graphs and its Boundary Vertices}
\author{Manoj Changat\\Department of Futures Studies, University of Kerala, Trivandrum\\mchangat@gmail.com\\Prasanth G.Narasimha-Shenoi\footnote{Research was supported by Science and Engineering Research Board, A Statutory board of Department of Science and Technology, Government of India under the grants
EMR/2015/002183.  Research was also supported by Kerala State Council for Science Technology and Environment of Government of Kerala under their SARD project grants Council(P) No. 436/2014/KSCSTE, dated 25/08/2014.} \& Mary Shalet T. J\samethanks\\Department of Mathematics\\Government College Chittur, Palakkad - 678104\\prasanthgns@gmail.com \& mary\_shallet@yahoo.co.in\\Ram Kumar\\Department of Mathematics\\ M.G.College Trivandrum\\ram.k.mail@gmail.com}

\begin{document}

\maketitle

\begin{abstract}
Suppose that $D=(V,E)$ is a strongly  connected digraph.  Let $u ,v\in  V(D)$.  The maximum distance $md (u,v)$ is defined as $md(u,v)$=max\{$\overrightarrow{d}(u,v), \overrightarrow{d}(v,u)$\} where $\overrightarrow{d}(u,v)$ denote the length of a shortest directed $u-v$ path in $D$.  This is a metric.  The boundary, contour, eccentric and peripheral sets of a strong digraph $D$ are defined with respect to this metric.  The main aim of this paper is to identify the above said metrically defined sets of a large strong digraph $D$ in terms of its prime factor decomposition with respect to cartesian product.

\end{abstract}

\section{Introduction}
In the present scenario, one way networks are frequently met across in all areas of day to day life.  But dealing with one way networks is much more difficult than two way networks.  As an instance, finding the distance between pairs of vertices in a one way network involves twice the number of steps involved in a two way network with same number of vertices.  Hence in complicated networks, the idea of prime factor decomposition
 have important applications.  The divide and conquer approach using prime factor decomposition helps to determine whether a given large digraph is strongly connected.  
 
 If the digraph is strong, we can apply the results obtained in this paper to 
 find the periphery,contour and eccentricity sets of large strongly connected digraphs.  This is accomplished by first applying any of the algorithms for finding the unique prime factor decomposition.If all except one of the factors have the two-sided eccentricity property,  then in order to find the periphery and contour sets , we need not examine the distance between all the vertices.  Instead, we need only to find the  distance between the vertices occurring in the same factor.    
 
 To see this, consider a strong digraph which has ten vertices. To find the periphery and contour, we have to find the eccentricity of the ten vertices, which involves ninety steps.  If it has a prime decomposition into two digraphs, one of them will have two vertices and the other five  vertices.  Thus if any one of these digraphs have the two-sided eccentricity property, we need only to find the  distance between the two vertices in the first digraph and
 the  distance between the five vertices in the second.  This involves two steps in the first digraph and twenty in the second, which adds up to total of 22 steps in the place of 90 steps in the original digraph.\\
 Thus it is evident that  as the number of vertices increase, we can save a considerable amount of work, provided  all except one of the factors have the two-sided eccentricity property.  
 
The one way problem have been studied since 1939 starting from the classical paper of Robbins \cite{robbins1939theorem}. A directed network is a network in which each edge has a direction, pointing from one vertex to another.
They have applications in a variety of different fields varying from computer science to theoretical biology \cite{newman2010networks}. The World Wide Web is a directed network with web pages as vertices and hyperlinks between pages as edges. The neural network consist of several neurons wired together and it is known that the brain constantly changes the pattern of wiring in response to inputs and experiences. In large networks similar to that of one way traffic, there arises the problem of designing the network so as to minimize the distance between nodes as well as to decrease the cost of construction of routes involved.

The boundary type vertices of a graph , the \textit{boundary, contour, eccentric} and \textit{peripheral sets} of a graph were studied in \cite{chartrand2003boundary} and \cite{caceres2005rebuilding}.  

The boundary type vertices of a graph can be roughly described as the vertices of a graph which constitute the borders of a graph.  All other vertices of the graph lie between them.  So they play a significant role in the theory of graphs.  

The distance $d(u,v)$ between two vertices $u$ and $v$ in a non trivial connected graph $G$ is the length of a shortest $u-v$ path in $G$.  For a vertex $v$ of $G$, the eccentricity $e(v)$ is the distance between $v$ and a vertex farthest from $v$. 

A vertex $v$ is said to be an \textit{eccentric vertex} of a vertex $u$ if $ecc(u)=d(u,v)$.  
A vertex $v$ is said to be a \textit{peripheral vertex} of G , if $ecc(v)=diam(G)$.  A vertex $v$ is said to be a \textit{boundary vertex} of a vertex   $u$ if for all neighbours $w$ of v, $d(u,w) \leq d(u,v)$. A vertex $v$ is said to be a \textit{contour vertex} of  $G$ if for all neighbours $w$ of $v$, $ecc(w) \leq ecc(v)$.  
 
Minimizing the distance between nodes in the digraph sense is equivalent to minimizing the distance in either direction. Thus the metric maximum distance $md(u,v)$ \cite{chartrand1997distance}, for $u ,v\in  V(D)$ find its application in these networks. We can extend the concept of boundary type vertices to the case of digraphs using the metric $md$.  The significance of the boundary type vertices lies in the fact that they determine the efficiency of a network. 

In the case of large networks, it is cumbersome to identify the various boundary type sets.  The problem is simplified if the network can be decomposed into smaller networks.
Several types of graph products have been studied and these can be extended to digraphs \cite{imrichklav}.  Cartesian product is the most important among the graph products and is widely used in metric graph theory.  Cartesian product of  graphs  was introduced by Gert Sabidussi  \cite{sabidussi1959graph}. Sabidussi showed that every connected undirected graph $G$ has a prime factorization  that is unique upto the order and isomorphisms of the factors.  After this, some faster factorization algorithms for undirected graphs were developed. Afterwards Feigenbaum proved that directed graphs have unique prime factorizations under cartesian multiplication and that we can find the prime factorizations of weakly connected
digraphs in polynomial time \cite{feigenbaum1986directed}. This was improved to a linear time approach by Crespelle et al\cite{crespelle2013linear}. Hence we attempt to derive some information about the above mentioned sets in terms of the factors in the prime decomposition .

\section{Preliminaries}
A \textit{directed graph} or \textit{digraph} $D$ is a triple consisting of a
vertex set $V(D)$, an edge set $E(D)$, and a function assigning each edge an
ordered pair of vertices.  The first vertex of the ordered pair is the tail of
the edge, and the second is the head; together they are the endpoints.  A
\textit{directed path} is a directed graph $P\neq \emptyset$  with distinct
vertices $u_0,\ldots , u_k$ and edges $e_0, \ldots , e_{k-1}$ such that $e_i$ is
an edge directed from $x_i$ to $x_{i+1}$, for all $i<k$.  In this paper  a
path will always mean `directed path'. A digraph is \textit{strongly
connected} or\textit{ strong} if for each ordered pair $u$, $v$ of vertices,
there is a path from $u$ to $v$.

The \textit{length} of a path is the number of its edges. Let $u$ and $v$ be vertices of a
strongly connected digraph $D$. A shortest directed $u- v$ path is also called a directed $u- v$ \textit{geodesic}. The number of edges in a  directed $u- v$ geodesic is called the directed distance $ \overrightarrow{d}(u,v)$ . 
But this distance is not a metric because $ \overrightarrow{d}(u,v)\neq \overrightarrow{d}(v,u)\}$.  So in \cite{chartrand1997distance},Chartrand and Tian introduced two other distances in a strong digraph, namely the maximum distance $md(u,v)=max\{\overrightarrow{d}(u,v),\overrightarrow{d}(v,u)\}$ and sum distance $sd(u,v)=\overrightarrow{d}(u,v)+\overrightarrow{d}(v,u)$, both of which are metrics.  In this paper, we deal with the first metric, the maximum distance $md$.  It is clear that the distance $md$ is positive and symmetric.For the sake of completion we will show $'md'$ satisfy the triangle inequality \cite{chartrand1997distance}:

Let $u,v,w\in V(D)$.  Suppose that $max\{\overrightarrow{d}(u,v),\overrightarrow{d}(v,u)\}=\overrightarrow{d}(u,v)$.

Then
\begin{align*}
md(u,v)&=max\{\overrightarrow{d}(u,v),\overrightarrow{d}(v,u)\}\\
&=\overrightarrow{d}(u,v)\\
&\leq \overrightarrow{d}(u,w)+\overrightarrow{d}(w,v)\\
&\leq max\{\overrightarrow{d}(u,w),\overrightarrow{d}(w,u)\}+max\{\overrightarrow{d}(w,v),\overrightarrow{d}(v,w)\}\\
&=md(u,w)+md(w,v)
\end{align*} 
Following \cite{nebesky2004directed}, we define the geodetic interval as follows:
$I(u,v)=\{w:md(u,w)+md(w,v)=md(u,v)\}$.  That is if $md(u,v)= \overrightarrow{d}(u,v)>\overrightarrow{d}(v,u)$, then $I[u,v]$ consists of only the vertices in the directed $u- v$ geodesics and not in the other direction.  For $S\subseteq V(D)$, the geodetic closure $I[S]$ of $S$ is the union of all geodetic intervals $I[u,v]$ over all pairs $u,v\in S$. So $I[S]=\cup_{u,v\in S} I[u,v]$.  In this paper we denote $md(u,v)$ by $d(u,v)$.  Most of the following definitions are analogous to the definitions in  \cite{chartrand2003boundary}.
Let $D$ be a strong digraph and $u, v \in V(D)$. The vertex $v$ is said to be a \textit{boundary vertex} of $u$ if no neighbor of $v$ is further away from $u$ than $v$.  A vertex $v$ is called a \textit{boundary vertex} of $D$ if it is the boundary vertex of
some vertex $u \in V(D)$.  The \textit{boundary} $\partial(D)$ of $D$ is the set of all of its boundary vertices; $\partial(D) = \{v \in V | \exists u \in V, \forall w \in N(v) : d(u,w)\leq d(u,
v)\}$.  Given a vertex set $W \subseteq V$, \textit{the eccentricity} in $W$ of a vertex $u \in W$ is defined as $ecc_W(u) = max\{d(u, v) | v \in W\}$.  In particular, if $W =V (G)$, then we write $ecc_W(u)=ecc_G(u)$, where
$ecc_G(u)=ecc(u)=max\{d(u, v) | v \in V \}$.  Given
$u, v \in V$, the vertex $v$ is called an \textit{eccentric vertex} of $u$ if no vertex in $V$ is further away from $u$ than $v$.  This means that $d(u, v) = ecc(u)$. A vertex $v$ is called an \textit{eccentric vertex} of $G$ if it is the eccentric vertex of some vertex $u \in V$.  The \textit{eccentricity} $Ecc(D)$ of $D$ is the set of all of its eccentric vertices;  
$Ecc(D) = \{v \in V | \exists u \in V, ecc(u) = d(u, v)\}$.  
A vertex $v \in V$ is called a \textit{peripheral vertex} of $D$ if no vertex in $V$ has an eccentricity greater than $ecc(v)$,  that is, if the
eccentricity of $v$ is exactly equal to the diameter $diam (D)$ of $D$.  The \textit{periphery} $Per(D)$ of $D$ is the set of all of its peripheral vertices;  
$Per(D) = \{v \in V | ecc(u)\leq ecc(v), \forall u \in V\} = \{v \in V | ecc(v) =
diam(D)\}$.  A vertex $v \in V$ is called a \textit{contour vertex} of $D$ if no neighbor
vertex of $v$ has an eccentricity greater than $ecc(v)$. The following
definition is from \cite{caceres2005rebuilding}.  The \textit{contour} $Ct(D)$ of $D$ is the set of all of its contour vertices;  
$Ct(D) = \{v \in V | ecc(u)\leq ecc(v), \forall u \in N(v)\}$.  
The following proproposition follow directly from the definitions.
\begin{proposition}
Let $D = (V ,E)$ be a strong digraph. Then, the following statements hold . \\
 1. $Per(D) \subseteq Ct(D)\cap Ecc(D)$.\\
2. $ Ecc(D) \cup Ct(D) \subseteq \partial(D)$.
\end{proposition}
In general, we can see that  the eccentricity of a vertex of a digraph with respect to the metric \textit{md} is one-sided, in the sense that the distance to the farthest vertex may occur only in one direction unlike the case of undirected graphs.  
So we make the following definition.  
\begin{definition}
A digraph $D$ is said to satisfy the two-sided eccentricity property, if for all $u_i\in D$	there exist vertices $u_j,u_k$ in $D$ (not necessarily distinct) such that 
	$ecc(u_i)=\overrightarrow{d}(u_i,u_j)=\overrightarrow{d}(u_k,u_i)$.
\end{definition}  
In \cite{caceres2006geodetic}, Caceres et al. proved the following proposition.
\begin{proposition}\label{proposition caceres}
Let $G=(V,E)$ be a connected graph.

\begin{enumerate}
	\item If $Ct(G)=Per(G)$, then $I[Ct(G)=V(G)$. 
	\item If $|Ct(G)|=|Per(G)|=2$, then either $|\partial(G)|=2$ or
	$|\partial(G)|\ge 4$.
	\item If $|Ecc(G)|=|Per(G)|+1$, then $|\partial(G)|> |Ecc(G)|$
	\item If $|Ecc(G)|>|Per(G)|$, then $|\partial(G)|\ge |Per(G)|+2$
\end{enumerate}
\end{proposition}

We checked whether the digraph analogue of proposition~\ref{proposition caceres} holds good with respect to the metric $md$.  It turned out that (1) and (2) need not hold.  Consider the digraph $D$ in example~\ref{exam1}.  Here $Ct(D)=Per(D)=Ecc(D)=\{w,z\}$ but $v \notin I\{w,z\}$.  This is because $\overrightarrow{d}(w,z)=3, \overrightarrow{d}(z,w)=2$ giving $d(w,z)=3$ whereas both $w-z$ directed paths passing through $v$ are of length $4$.  \\Also $\partial(D)=\{v,w,z\}$ as $v$ is a boundary vertex of $w$ while $x$ and $y$ are not boundary vertices of any vertex.
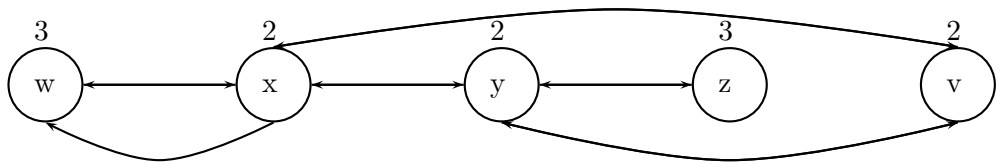
\begin{example}\label{exam1}
\begin{center}
\begin{figure}[H]
\begin{pspicture}(2,4)
\pscircle(0,3){.5}
\put(-.15,2.9)w
\put(-.15,3.6)3
\pscircle(3,3){.5}
\put(2.85,2.9)x
\put(2.85,3.6)2
\psline{->}(0.5,3)(2.5,3)
\psline{->}(2.5,3)(0.5,3)
\pscircle(6,3){.5}
\put(5.85,2.9)y
\put(5.85,3.6)2
\psline{->}(3.5,3)(5.5,3)
\psline{->}(5.5,3)(3.5,3)
\pscircle(9,3){.5}
\put(8.85,2.9)z
\put(8.85,3.6)3
\psline{->}(6.5,3)(8.5,3)
\psline{->}(8.5,3)(6.5,3)
\pscurve{->}(3,2.5)(1.5,2)(0,2.5)
\pscircle(12,3){.5}
\put(11.85,2.9)v
\put(11.85,3.6)2
\pscurve{->}(12,2.5)(9,2)(6,2.5)

\pscurve{->}(6,2.5)(9,2)(12,2.5)

\pscurve{->}(12,3.5)(7.5,4)(3,3.5)

\pscurve{->}(3,3.5)(7.5,4)(12,3.5)

\end{pspicture}
\end{figure}
\end{center}
\end{example}
The above said variation of digraphs from undirected graphs motivated us to investigate various other results related to the boundary type sets of undirected graphs in the case of digraphs.  

Even though the proofs of (3) and (4) of the proposition~\ref{proposition caceres} follow the same lines of proof of  proposition~\ref{proposition caceres} as in \cite{caceres2006geodetic}, for the sake of completeness, we give the proofs below.
\begin{proof}
Let $x\in Ecc(D)-Per(D)$.\\ 
Take $W=\{y\in V(D)|d(y,x)=ecc(y)\}$. Then $W\cap Per(D)=\emptyset$,  since $x\notin Per(D)$.  Also $W\cap Ecc(D)=\emptyset$, since $Ecc(D)=Per(D)\cup \{x\}$.  Consider a vertex $z\in W $ such that $\displaystyle{\text{ecc}(z)=\max_{y\in W} ecc(y)}$. \\To prove that $z$ is a boundary vertex of $x$, let us assume to the contrary that there exists $w\in N(z)$ such that $d(w,x)=d(z,x)+1$  which gives $ecc(w)=ecc(z)+1$ and $w\in W$, which is a contradiction.  Hence $z\in \partial(D)$. 
Given that $Per(D)\subset Ecc(D)$. \\
Also $Ecc(D)\subseteq \partial(D)$.  Hence as in the previous proof, $|\partial(D)|>|Ecc(D)|>|Per(D)|$ which gives $|\partial(D)|\geq|Per(D)|+2$. 
\end{proof}

\section{Cartesian product of directed graphs}

The Cartesian product of two directed graphs $ D_1 = (V_1, E_1)$ and
$D_2 = (V_2, E_2)$ is a digraph $D$ with vertex set $V(D) = V_1 \times V_2$ in which vertices
$(u_i, v_r) \text{ is adjacent to } (u_j,v_s)$ if either $u_i = u_j$ and
$(v_r,  v_s)\in E_2$ or $v_r =v_s$ and $(u_i, u_j)\in E_1$.  It is denoted by
$D_1\square D_2$.  In a similar manner, we can define the cartesian product of n directed graphs, $D_1\square D_2 \ldots\square D_n$.\\  The cartesian product $D =D_1\square D_2 \ldots\square D_n$ of $n$ directed graphs is the directed graph $D = (V (D), E(D))$
whose vertex set is $\displaystyle{V(D) = \Pi_{1 \leq i\leq p} V (D_i )}$ and such that for all $x, y \in V (D)$,
with $x = (x_1 , \ldots, x_n )$ and $y = (y_1 ,\ldots, y_n )$, we have $x$ is adjacent to $y$ if and only
if there exists $i \in \{1,2, \ldots n\}$ such that for all $j \neq i, x_j = y_j$ and $x_i$ is adjacent to $y_i$ in $D_i$.  As in the case of undirected graphs, cartesian multiplication is commutative  and associative in the case of directed graphs also.  A digraph $D$ is prime with regard to the
cartesian product if and only if for all digraphs $D_1 , D_2$ such that $D=D_1 \square D_2$ then $D_1$ or $D_2$ has only one vertex.  Now we state the Fundamental Theorem of cartesian products (Unicity of the prime decomposition of digraphs\cite{feigenbaum1986directed}.  For any weakly connected directed graph $D$, there exists a unique $n \ge 1$ and a unique tuple $(D_1 , \ldots , D_n )$ of digraphs up to reordering and isomorphism of the factors $D_i$, such that each $D_i$ has at least two vertices, each $D_i$ is prime for the cartesian product and $D =D_1\square D_2 \ldots\square D_n$. $(D_1 , \ldots , D_n)$ is called the prime decomposition of $D$.  The following proposition can be seen in \cite{harary1966connectedness} and we give a short proof for that also.  
\begin{proposition}\cite{harary1966connectedness} $D_1\square D_2$ is strongly
connected if and only if both $D_1$ and $D_2$ are strongly connected. 
\end{proposition}
\begin{proof}
Necessary part:\\
Let $D_1\square D_2$ be strongly connected. If any one of $D_1$ or $D_2$, say $D_1 $ is not  strongly connected, there exist two vertices $u_i, u_j$ in $D_1$ such that there is no directed path from $u_i$ to $u_j$.
Hence there exist no directed path from $(u_i,  v_r)$ to $(u_j,  v_r),\forall v_r\in V(D_2)$ ,which is a contradiction.  Hence both $D_1$ and $D_2$ are strongly connected.\\
Sufficient part:\\
Let $D_1$ and $D_2$ be strongly connected. Consider two arbitrary vertices $(u_i,  v_r)$ and $(u_j,  v_s)\in V(D_1\square D_2)$.Then since $D_1$ and $D_2$ are strongly connected, there exist directed paths in both directions between $u_i$ and $u_j$ in $D_1$ and between $v_r$ and $v_s$ in $D_2$.  Hence there exist directed paths from $(u_i,  v_r)$ to $(u_j,  v_r)$ and $(u_j,  v_r)$ to $(u_j,  v_s)$ in $D_1\square D_2$. Combining these paths ,we get a directed path from $(u_i,  v_r)$ to $(u_j,  v_s)$. Similarly, we get a directed path from $(u_j,  v_s)$ to $(u_i,  v_r)$.  Thus $D_1\square D_2$ is strongly connected.
\end{proof}
We have an immediate corollary.
\begin{corollary}
$D_1\square D_2 \ldots\square D_n$ is strongly
	connected if and only if  $D_1 , D_2 , \ldots D_n$ are strongly connected. 
\end{corollary}
\subsection{Distance between two vertices}
\begin{lemma}
Let $ D_1$ and $D_2$ be two strongly connected digraphs with vertex sets $\{u_1,u_2,\ldots u_m\}$ and $\{v_1,v_2,\ldots v_n\}$ respectively. Then
$d((u_i,v_r),(u_j,v_s))=max\{\overrightarrow{d}(u_i,u_j)+\overrightarrow{d}(v_r
,v_s),\overrightarrow{d}(u_j, u_i)+\overrightarrow{d}(v_s,v_r)\}$.
\end{lemma}
\begin{proof}

$d((u_i,v_r),(u_j,v_s))=max\{\overrightarrow{d}((u_i,v_r),(u_j,v_s)),
\overrightarrow{d}((u_j,v_s),(u_i,v_r))\}$

The shortest path in the direction from $(u_i,v_r)$ to $(u_j,v_s)$ is either
the directed path from $(u_i,v_r)$ to $(u_i,v_s)$ and then from $(u_i,v_s)$
to $(u_j,v_s)$ or from $(u_i,v_r)$ to $(u_j,v_r)$  and then from  $(u_j,v_r)$
to $(u_j,v_s)$.  In both the
cases, $\overrightarrow{d}((u_i,v_r),(u_j,v_s))=\overrightarrow{d}(u_i,
u_j)+\overrightarrow{d}(v_r,
v_s)$.  Similarly $\overrightarrow{d}((u_j,v_s),(u_i,v_r))=\overrightarrow{d}(u_j,
u_i)+\overrightarrow{d}(v_s,v_r)$.  Therefore 
$d((u_i,v_r),(u_j,v_s))=max\{\overrightarrow{d}(u_i,
u_j)+\overrightarrow{d}(v_r,
v_s),\overrightarrow{d}(u_j, u_i)+\overrightarrow{d}(v_s,v_r)\}$.
\end{proof}
See the following example . 
\begin{example}\label{exam2}
\begin{figure}[H]
\centering
\begin{pspicture}(2,8)
\pscircle(0,7){.25}
\put(-.15,6.9){$u_1$}
\put(-.15,7.25){$2$}

\pscircle(2,7){.25}
\put(1.85,6.9){$u_2$}
\put(1.85,7.25)2
\pscircle(4,7){.25}
\put(3.85,6.9){$u_3$}
\put(3.85,7.25)2
\psline{->}(.4,7)(1.6,7)
\psline{->}(2.4,7)(3.6,7)
\pscurve{->}(4,6.6)(2,6)(0,6.6)
\put(1.9,5){$D_1$}

\pscircle(5,7){.25}
\put(4.85,6.9){$v_1$}
\put(5.25,6.9)2
\pscircle(5,5){.25}
\put(4.85,4.9){$v_2$}
\put(5.25,4.9)2
\pscircle(5,3){.25}
\put(4.85,2.9){$v_3$}
\put(5.25,2.9)2
\psline{->}(5,6.6)(5,5.4)
\psline{->}(5,4.6)(5,3.4)
\pscurve{->}(5.2,3.2)(6,5)(5.4,6.6)
\put(4.9,1.9){$D_2$}

\pscircle(2,1){.2}
\put(1.95,1.25){$(u_1,v_1)$}
\put(1.65,1.25)4
\pscircle(4,1){.2}
\put(3.95,1.25){$(u_2,v_1)$}
\put(3.65,1.25)4
\pscircle(6,1){.2}
\put(5.95,1.25){$(u_3,v_1)$}
\put(5.65,1.25)4
\psline{->}(2.2,1)(3.8,1)
\psline{->}(4.2,1)(5.8,1)
\pscurve{->}(6,0.8)(4,0.5)(2,0.8)

\pscircle(2,-1){.2}
\put(1.95,-.75){$(u_1,v_2)$}
\put(1.65,-.75)4
\pscircle(4,-1){.2}
\put(3.95,-.75){$(u_2,v_2)$}
\put(3.65,-.75)4
\pscircle(6,-1){.2}
\put(5.95,-.75){$(u_3,v_2)$}
\put(5.65,-.75)4
\psline{->}(2.2,-1)(3.8,-1)
\psline{->}(4.2,-1)(5.8,-1)
\pscurve{->}(6,-1.2)(4,-1.5)(2,-1.2)
\put(3.6,-4.2){$D_1\square D_2$}

\pscircle(2,-3){.2}
\put(1.95,-2.75){$(u_1,v_3)$}
\put(1.65,-2.75)4
\pscircle(4,-3){.2}
\put(3.95,-2.75){$(u_2,v_3)$}
\put(3.65,-2.75)4
\pscircle(6,-3){.2}
\put(5.95,-2.75){$(u_3,v_3)$}
\put(5.65,-2.75)4 
\psline{->}(2.2,-3)(3.8,-3)
\psline{->}(4.2,-3)(5.8,-3)
\pscurve{->}(6,-3.2)(4,-3.5)(2,-3.2)
\psline{->}(1.95,.8)(1.95,-.8)
\psline{->}(1.95,-1.2)(1.95,-2.8)
\pscurve{->}(1.95,-2.8)(1.5,-1)(1.95,.8)

\psline{->}(3.95,.8)(3.95,-.8)
\psline{->}(3.95,-1.2)(3.95,-2.8)
\pscurve{->}(3.95,-2.8)(3.5,-1.)(3.95,.8)

\psline{->}(5.95,.8)(5.95,-.8)
\psline{->}(5.95,-1.2)(5.95,-2.8)
\pscurve{->}(5.95,-2.8)(5.5,-1)(5.95,.8)

\end{pspicture}
\end{figure}
\end{example}
  \vspace{4cm}
\subsection{Comparing with the graph case}
We can see that in general it does not satisfy
$d((u_i,v_r),(u_j,v_s))=d(u_i,u_j)+d(v_r,v_s)$, which is true in the case of
cartesian product of two simple graphs.  Consider example ~\ref {exam2}.  $d((u_1,v_3),(u_3,v_1))=max\{\overrightarrow{d}(u_1,u_3)+\overrightarrow{d}(v_3
,v_1),\overrightarrow{d}(u_3, u_1)+\overrightarrow{d}(v_1,v_3)\}
$=$max\{2+1,1+2\}=3\neq d(u_1,u_3)+d(v_3,v_1)$.  Consequently, $ecc(u_i,v_r)\neq ecc(u_i)+ecc(v_r)$ unlike in the graph case. But we can show that $d((u_i,v_r),(u_j,v_s))\leq d(u_i,u_j)+d(v_r,v_s)$
and $ecc_{D_1\square D_2}(u_i,v_r)\leq ecc_{D_1}(u_i)+ecc_{D_2}(v_r)$.  
\begin{theorem}
Let $D_1$ and $D_2$ be two strongly connected digraphs.  Then
\begin{equation*}
d((u_i,v_r),(u_j,v_s))\leq d(u_i,u_j)+d(v_r,v_s)
\end{equation*}
 for all
$(u_i,v_r),(u_j,v_s) \in V(D_1\square D_2)$.
\end{theorem}
\begin{proof}
$d((u_i,v_r),(u_j,v_s))=max\{\overrightarrow{d}(u_i,
u_j)+\overrightarrow{d}(v_r,
v_s),\overrightarrow{d}(u_j, u_i)+\overrightarrow{d}(v_s,v_r)\} $

We have $d(u_i,u_j)=max\{\overrightarrow{d}(u_i,u_j),\overrightarrow{d}(u_j,
u_i)\}$ and 
$d(v_r,v_s)=max\{\overrightarrow{d}(v_r,v_s),\overrightarrow{d}(v_s,
v_r)\}$

We have 4 cases:

Case 1:$d(u_i,u_j)=\overrightarrow{d}(u_i,u_j)$ and 
$d(v_r,v_s)=\overrightarrow{d}(v_r,v_s)$.\\

Case 2:$d(u_i,u_j)=\overrightarrow{d}(u_j,u_i)$ and
$d(v_r,v_s)=\overrightarrow{d}(v_s,v_r)$.\\

Case 3:$d(u_i,u_j)=\overrightarrow{d}(u_i,u_j)$ and 
$d(v_r,v_s)=\overrightarrow{d}(v_s,v_r)$.\\

Case 4:$d(u_i,u_j)=\overrightarrow{d}(u_j,u_i)$ and
$d(v_r,v_s)=\overrightarrow{d}(v_r,v_s)$.

In all these cases we get $\overrightarrow{d}(u_i,
u_j)+\overrightarrow{d}(v_r,v_s)\leq d(u_i,u_j)+d(v_r,v_s)$ and
$\overrightarrow{d}(u_j,
u_i)+\overrightarrow{d}(v_s,v_r)\leq d(u_i,u_j)+d(v_r,v_s)$.  Therefore
$max\{\overrightarrow{d}(u_i,u_j)+\overrightarrow{d}(v_r,
v_s),\overrightarrow{d}(u_j, u_i)+\overrightarrow{d}(v_s,v_r)\} 
\leq d(u_i,u_j)+d(v_r,v_s)$.  So $d((u_i,v_r),(u_j,v_s))\leq d(u_i,u_j)+d(v_r,v_s)$.
\end{proof}
 \begin{corollary}
 $ecc_{D_1\square D_2}(u_i,v_r)\leq ecc_{D_1}(u_i)+ecc_{D_2}(v_r)$.
 \end{corollary}
\begin{proof}
Let $(u_j,v_s)$ be an eccentric vertex of $(u_i,v_r)$ in $D_1\square D_2$.Then $ecc_{D_1\square D_2}(u_i,v_r)=d((u_i,v_r),(u_j,v_s))\leq d(u_i,u_j)+d(v_r,v_s)\leq ecc_{D_1}(u_i)+ecc_{D_2}(v_r)$.

\end{proof} 
	\section{Some Remarks}
	\begin{remark}
$u$ is an eccentric vertex of $u'$ in  $D_1$ and $v$ is an eccentric vertex of $v'$ need not imply that $(u,v)$ is an eccentric vertex of $(u',v')$ in $D_1\square D_2$. 
	\end{remark}
Consider the digraph in   example ~\ref {exam2}.  We can see that $u_1$ is an eccentric vertex of $u_3$ in $D_1$ and $v_3$ is an eccentric vertex of $v_1$ in $D_2$.  But $(u_1,v_3)$ is not an eccentric vertex of 
$(u_3,v_1)$ as 
\begin{align*}
d((u_3,v_1),(u_1,v_3))&=\max\{\overrightarrow{d}(u_3,
u_1)+\overrightarrow{d}(v_1,v_3),\overrightarrow{d}(u_1,
u_3)+\overrightarrow{d}(v_3,v_1)\}\\&=max\{1+2,2+1\}\\&=3\\&<d((u_3,v_1),(u_1,v_2))\\&=max\{\overrightarrow{d}(u_3,
u_1)+\overrightarrow{d}(v_1,v_2),\overrightarrow{d}(u_1,
u_3)+\overrightarrow{d}(v_2,v_1)\}\\&=max\{1+1,2+2\}\\&=4.
\end{align*}  Here $(u_1,v_2)$ is an eccentric vertex of $(u_3,v_1)$ in $D_1\square D_2$ and $(u_1,v_3)$ is an eccentric vertex of $(u_3,v_2)$.

\begin{remark}
A vertex $(u,v)$ can be an eccentric vertex of $(u',v')$ in $D_1\square D_2$ without $u$ being an eccentric vertex of $u'$ in $D_1$ or $v$ being an eccentric vertex of $v'$ in $D_2$ .
	\end{remark}
Consider the digraphs given in example ~\ref{remark}. We have 
\begin{align*}
d((u_1,v_1),(u_4,v_5))&=max\{\overrightarrow{d}(u_1,u_4)+\overrightarrow{d}(v_1,
v_5),\overrightarrow{d}(u_4, u_1)+\overrightarrow{d}(v_5,v_1)\}\\&=max\{3+2,3+4\}\\&=7\\&=ecc(u_1,v_1).
\end{align*} So $(u_4,v_5)$ is an eccentric vertex of $(u_1,v_1)$ in $D_1\square D_2$ whereas $u_4$ is not an eccentric vertex of $u_1$ in $D_1$.

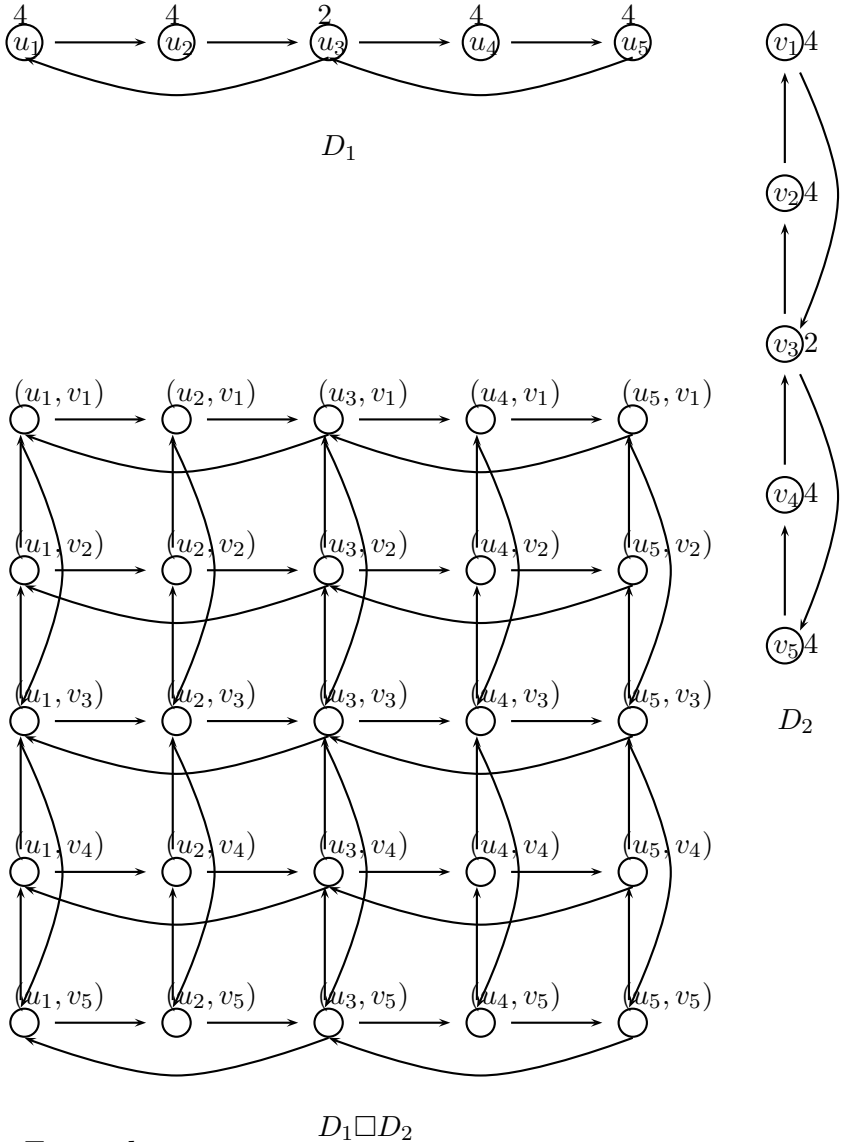
\begin{center}
\begin{example}\label{remark}
	
\begin{figure}[H]
\begin{pspicture}(2,15)

\pscircle(0,14){.25}
\put(-.15,13.9){$u_1$}
\put(-.15,14.25)4

\pscircle(2,14){.25}
\put(1.85,13.9){$u_2$}
\put(1.85,14.25)4
\pscircle(4,14){.25}
\put(3.85,13.9){$u_3$}
\put(3.85,14.25)2
\psline{->}(.4,14)(1.6,14)
\psline{->}(2.4,14)(3.6,14)
\pscurve{->}(4,13.8)(2,13.3)(0,13.8)
\pscircle(6,14){.25}
\put(5.85,13.9){$u_4$}
\put(5.85,14.25)4
\pscircle(8,14){.25}
\put(7.85,13.9){$u_5$}
\put(7.85,14.25)4
\psline{->}(4.4,14)(5.6,14)
\psline{->}(6.4,14)(7.6,14)
\pscurve{->}(8,13.8)(6,13.3)(4,13.8)
\put(3.9,12.5){$D_1$}

\pscircle(10,14){.25}
\put(9.85,13.9){$v_1$}
\put(10.25,13.9)4
\pscircle(10,12){.25}
\put(9.85,11.9){$v_2$}
\put(10.25,11.9)4
\pscircle(10,10){.25}
\put(9.85,9.9){$v_3$}
\put(10.25,9.9)2
\psline{->}(10,12.4)(10,13.6)
\psline{->}(10,10.4)(10,11.6)
\pscurve{->}(10.2,13.6)(10.7,12)(10.2,10.2)
\pscircle(10,8){.25}
\put(9.85,7.9){$v_4$}
\put(10.25,7.9)4
\pscircle(10,6){.25}
\put(9.85,5.9){$v_5$}
\put(10.25,5.9)4
\psline{->}(10,8.4)(10,9.6)
\psline{->}(10,6.4)(10,7.6)
\pscurve{->}(10.2,9.6)(10.7,8)(10.2,6.2)

\put(9.9,4.9){$D_2$}

\pscircle(0,9){.2}
\put(-.15,9.25){$(u_1,v_1)$}
\pscircle(2,9){.2}
\put(1.85,9.25){$(u_2,v_1)$}
\pscircle(4,9){.2}
\put(3.85,9.25){$(u_3,v_1)$}
\pscircle(6,9){.2}
\pscircle(8,9){.2}
\put(5.85,9.25){$(u_4,v_1)$}
\put(7.85,9.25){$(u_5,v_1)$}
\psline{->}(0.4,9)(1.6,9)
\psline{->}(2.4,9)(3.6,9)
\psline{->}(4.4,9)(5.6,9)
\psline{->}(6.4,9)(7.6,9)
\pscurve{->}(4,8.8)(2,8.3)(0,8.8)
\pscurve{->}(8,8.8)(6,8.3)(4,8.8)

\pscircle(0,7){.2}
\put(-.15,7.25){$(u_1,v_2)$}
\pscircle(2,7){.2}
\put(1.85,7.25){$(u_2,v_2)$}
\pscircle(4,7){.2}
\put(3.85,7.25){$(u_3,v_2)$}
\pscircle(6,7){.2}
\pscircle(8,7){.2}
\put(5.85,7.25){$(u_4,v_2)$}
\put(7.85,7.25){$(u_5,v_2)$}
\psline{->}(0.4,7)(1.6,7)
\psline{->}(2.4,7)(3.6,7)
\psline{->}(4.4,7)(5.6,7)
\psline{->}(6.4,7)(7.6,7)
\pscurve{->}(4,6.8)(2,6.3)(0,6.8)
\pscurve{->}(8,6.8)(6,6.3)(4,6.8)

\pscircle(0,5){.2}
\put(-.15,5.25){$(u_1,v_3)$}
\pscircle(2,5){.2}
\put(1.85,5.25){$(u_2,v_3)$}
\pscircle(4,5){.2}
\put(3.85,5.25){$(u_3,v_3)$}
\pscircle(6,5){.2}
\pscircle(8,5){.2}
\put(5.85,5.25){$(u_4,v_3)$}
\put(7.85,5.25){$(u_5,v_3)$}
\psline{->}(0.4,5)(1.6,5)
\psline{->}(2.4,5)(3.6,5)
\psline{->}(4.4,5)(5.6,5)
\psline{->}(6.4,5)(7.6,5)
\pscurve{->}(4,4.8)(2,4.3)(0,4.8)
\pscurve{->}(8,4.8)(6,4.3)(4,4.8)

\pscircle(0,3){.2}
\put(-.15,3.25){$(u_1,v_4)$}
\pscircle(2,3){.2}
\put(1.85,3.25){$(u_2,v_4)$}
\pscircle(4,3){.2}
\put(3.85,3.25){$(u_3,v_4)$}
\pscircle(6,3){.2}
\pscircle(8,3){.2}
\put(5.85,3.25){$(u_4,v_4)$}
\put(7.85,3.25){$(u_5,v_4)$}
\psline{->}(0.4,3)(1.6,3)
\psline{->}(2.4,3)(3.6,3)
\psline{->}(4.4,3)(5.6,3)
\psline{->}(6.4,3)(7.6,3)
\pscurve{->}(4,2.8)(2,2.3)(0,2.8)
\pscurve{->}(8,2.8)(6,2.3)(4,2.8)

\pscircle(0,1){.2}
\put(-.15,1.25){$(u_1,v_5)$}
\pscircle(2,1){.2}
\put(1.85,1.25){$(u_2,v_5)$}
\pscircle(4,1){.2}
\put(3.85,1.25){$(u_3,v_5)$}
\pscircle(6,1){.2}
\pscircle(8,1){.2}
\put(5.85,1.25){$(u_4,v_5)$}
\put(7.85,1.25){$(u_5,v_5)$}
\psline{->}(0.4,1)(1.6,1)
\psline{->}(2.4,1)(3.6,1)
\psline{->}(4.4,1)(5.6,1)
\psline{->}(6.4,1)(7.6,1)
\pscurve{->}(4,.8)(2,0.3)(0,.8)
\pscurve{->}(8,.8)(6,0.3)(4,.8)

\psline{->}(-.05,7.3)(-.05,8.8)
\psline{->}(1.95,7.3)(1.95,8.8)
\psline{->}(3.95,7.3)(3.95,8.8)
\psline{->}(5.95,7.3)(5.95,8.8)
\psline{->}(7.95,7.3)(7.95,8.8)
\psline{->}(-.05,5.3)(-.05,6.8)
\psline{->}(1.95,5.3)(1.95,6.8)
\psline{->}(3.95,5.3)(3.95,6.8)
\psline{->}(5.95,5.3)(5.95,6.8)
\psline{->}(7.95,5.3)(7.95,6.8)
\pscurve{->}(-.05,8.7)(.5,7)(-.05,5.2)
\pscurve{->}(1.95,8.7)(2.5,7)(1.95,5.2)
\pscurve{->}(3.95,8.7)(4.5,7)(3.95,5.2)
\pscurve{->}(5.95,8.7)(6.5,7)(5.95,5.2)
\pscurve{->}(7.95,8.7)(8.5,7)(7.95,5.2)
\psline{->}(-.05,3.3)(-.05,4.8)
\psline{->}(1.95,3.3)(1.95,4.8)
\psline{->}(3.95,3.3)(3.95,4.8)
\psline{->}(5.95,3.3)(5.95,4.8)
\psline{->}(7.95,3.3)(7.95,4.8)
\psline{->}(-.05,1.3)(-.05,2.8)
\psline{->}(1.95,1.3)(1.95,2.8)
\psline{->}(3.95,1.3)(3.95,2.8)
\psline{->}(5.95,1.3)(5.95,2.8)
\psline{->}(7.95,1.3)(7.95,2.8)
\pscurve{->}(-.05,4.7)(.5,3)(-.05,1.2)
\pscurve{->}(1.95,4.7)(2.5,3)(1.95,1.2)
\pscurve{->}(3.95,4.7)(4.5,3)(3.95,1.2)
\pscurve{->}(5.95,4.7)(6.5,3)(5.95,1.2)
\pscurve{->}(7.95,4.7)(8.5,3)(7.95,1.2)

\put(3.85,-.5){$D_1\square D_2$}
\end{pspicture}
\end{figure}
\end{example}
\end{center}
Another interesting remark is on Peripheral vertices.
\begin{remark}\label{eccentricity}
If $u$ is a peripheral vertex in $D_1$ and $v$ is a peripheral vertex in $D_2$ need not imply that $(u,v)$ is a peripheral vertex in $D_1\square D_2$.
\end{remark}
Consider the example~\ref{remark}.  $u_1$ is a peripheral vertex in $D_1$ and $v_1$ is a peripheral vertex in $D_2$. But  $(u_1,v_1)$ is not a peripheral vertex in $D_1\square D_2$.  Since $ecc(u_1,v_1)=7$ whereas $ecc(u_5,v_5)=8$.  Next we give a sufficient condition for the remark~\ref{eccentricity}.
\begin{proposition}
	Let $ D_1$ and $D_2$ be two strongly connected digraphs.  A sufficient condition for a vertex $(u_i,v_r)$ to satisfy
	$ecc_{D_1\Box D_2}(u_i,v_r)=ecc_{D_1}(u_i)+ecc_{D_2}(v_r)$ is that either $ D_1$ or $D_2$ satisfy the two-sided eccentricity property.  
\end{proposition}
\begin{proof}
	Suppose that either $ D_1$ or $D_2$ satisfy the two-sided eccentricity property.  
	Hence either there exist vertices $u_j,u_k$ in $D_1$ ($u_j$ may be equal to $u_k$)
	such that
	\begin{eqnarray}\label{suf}
	ecc(u_i)=\overrightarrow{d}(u_i,u_j)=\overrightarrow{d}(u_k,u_i) 
	\end{eqnarray} or there exist vertices $v_q,v_s$ in $D_2$ ($v_q$ may be equal to
	$v_s$) such that
	\begin{eqnarray}\label{suf1}
	ecc(v_r)=\overrightarrow{d}(v_r,v_q)=\overrightarrow{d}(v_s,v_r) 
	\end{eqnarray}
	Without loss of generality, suppose that condition ~\ref{suf} is satisfied in $D_1$ and  
	$u_j,u_k$ are the eccentric vertices of $u_i$. \\
	Case 1: Suppose $u_j\neq u_k$.\\
	Let $ecc(u_i)=\ell$.  Then
	$\overrightarrow{d}(u_i,u_j)=\overrightarrow{d}(u_k,u_i)=\ell$.  Let $v_r\in
	V(D_2)$ and $v_s$ be an eccentric vertex of $v_r$.  Let $ecc(v_r)={\ell}'$.  So either $\overrightarrow{d}(v_r,v_s)={\ell}'$ and $\overrightarrow{d}(v_s,v_r)<{\ell}'$ or
	$\overrightarrow{d}(v_r,v_s)<{\ell}'$ and $\overrightarrow{d}(v_s,v_r)={\ell}'$ or
	$\overrightarrow{d}(v_r,v_s)=\overrightarrow{d}(v_s,v_r)={\ell}'$.  Now consider
	$(u_i,v_r)\in V(D_1\Box D_2)$.\\
	Subcase 1.1: $\overrightarrow{d}(v_r,v_s)=\ell'$ and
	$\overrightarrow{d}(v_s,v_r)<\ell'$\\
	$d((u_i,v_r), (u_j,v_s))=\text{ max }
	\{\overrightarrow{d}(u_i,u_j)+\overrightarrow{d}(v_r,v_s),\overrightarrow{d}(u_j
	,u_i)+\overrightarrow{d}(v_s,v_r)\}=\ell+\ell '=ecc(u_i)+ecc(v_r)$.\\
	Subcase 1.2: $\overrightarrow{d}(v_r,v_s)<\ell'$ and
	$\overrightarrow{d}(v_s,v_r)=\ell'$\\
	$d((u_i,v_r), (u_k,v_s))=\text{ max }
	\{\overrightarrow{d}(u_i,u_k)+\overrightarrow{d}(v_r,v_s),\overrightarrow{d}(u_k
	,u_i)+\overrightarrow{d}(v_s,v_r)\}=\ell+\ell '=ecc(u_i)+ecc(v_r)$.\\
	Subcase 1.3: $\overrightarrow{d}(v_r,v_s)=\overrightarrow{d}(v_s,v_r)=\ell'$\\
	Then both $d((u_i,v_r), (u_k,v_s))=d((u_i,v_r),
	(u_j,v_s))=\ell+\ell'=ecc(u_i)+ecc(v_r)$.\\
	Case 2: Suppose that $u_j= u_k$.\\
	$\overrightarrow{d}(u_i,u_j)=\overrightarrow{d}(u_j,u_i)=\ell$.  As in the
	subcases of case 1, $d((u_i,v_r), (u_j,v_s))=ecc(u_i)+ecc(v_r)$.  From the
	result $ecc_{D_1\Box D_2}(u_i,v_r)\leq ecc_{D_1}(u_i)+ecc_{D_2}(v_r)$, the
	result follows.
\end{proof}

\begin{remark}
The above condition is not necessary for a vertex to satisfy 
	
	$ecc_{D_1\Box D_2}(u_i,v_r)=ecc_{D_1}(u_i)+ecc_{D_2}(v_r)$.  	
\end{remark}
In example ~\ref{remark},  $ecc_{D_1\Box D_2}(u_1,v_1)=ecc_{D_1}(u_1)+ecc_{D_2}(v_1)$	
even though none of the conditions ~\ref{suf} and ~\ref{suf1} are satisfied.

Boundary type sets of Cartesian product of two undirected graphs have many interesting properties and have been studied by Bresar et.al\cite{brevsar2008geodetic}.  It was proved that for any graphs $G$ and $H$,
\begin{theorem}\cite{brevsar2008geodetic}\label{boundary}
\begin{enumerate}
\item
$\partial(G\square H) = \partial(G)\times \partial(H) $
\item
$Ct(G\square H)= Ct(G)\times Ct(H) $
\item
$Ecc(G\square H)= Ecc(G)\times Ct(H) $
\item
$Per(G\square H)= Per(G)\times Per(H) $
\end{enumerate}

\end{theorem}

Corresponding to the theorem~\ref{boundary}, here we obtain the following results.  

\begin{theorem}
For any two strongly connected digraphs $D_1$ and $D_2$,
\begin{enumerate}
	\item
$Per(D_1\square D_2)\subseteq Per(D_1)\times Per(D_2) $ 
	\item

$Ct(D_1\square D_2)\subseteq Ct(D_1)\times Ct(D_2) $ 
\end{enumerate}
\end{theorem}
\begin{proof}
	\begin{enumerate}
		\item
Let $(u_j,v_s)\in Per(D_1\square D_2)$.  Then 
\begin{align*}
ecc(u_j,v_s)&=diam(D_1\square D_2)\\&=\displaystyle{max_{u_i \in V(D_1)}\{ecc(u_i)\}+max_{v_r \in V(D_2)}\{ecc(v_r)\}}.
\end{align*}
So $ecc(u_j)=\displaystyle{max_{u_i \in V(D_1)}\{ecc(u_i)\}}$ and $ecc(v_s)=max_{v_r \in V(D_2)}\{ecc(v_r)\}$.\\
Therefore $u_j\in Per(D_1)$ and $v_s\in Per(D_2)$. Hence the result.
\item
Let $(u_i,v_r) \in Ct(D_1\square D_2)$. If possible, let $u_i \notin Ct(D_1)$.  Then there is a vertex $u_j \in N(u_i) $ such that $ecc(u_j) > ecc(u_i)$.  Let $ecc(u_j)=ecc(u_i)+ \ell$.  Then by construction of $D_1\square D_2$ we get $ecc(u_j,v_r)=ecc(u_i,v_r)+\ell$ which is a contradiction since $(u_j,v_r) \in N(u_i,v_r) $.  Similarly we can show that $v_r \in Ct(D_2)$.  Hence $Ct(D_1\square D_2)\subseteq Ct(D_1)\times Ct(D_2) $.  
\end{enumerate}
\end{proof}

But in general we can show that   
\begin{enumerate}
	\item
	$Per(D_1)\times Per(D_2)\nsubseteq Per(D_1\square D_2)$
\item
$Ct(D_1)\times Ct(D_2)\nsubseteq Ct(D_1\square D_2)$

\item
$Ecc(D_1)\times Ecc(D_2)\nsubseteq Ecc(D_1\square D_2) $
\item
  $\partial(D_1)\times \partial(D_2)\nsubseteq \partial(D_1\square D_2) $
\end{enumerate}  To establish this, consider the digraph in  example ~\ref{exam3}.  Here\\
 $Per(D_1) = Ct(D_1) = Ecc(D_1) = \partial(D_1) = \{u_1,u_2,u_3\}$ and  $ Per(D_2) = Ct(D_2) = Ecc(D_2) = \partial(D_2) = \{v_1,v_2,v_3\}$.  But we can see that  $(u_1,v_1) \notin Per(D_1\square D_2)$ and 
 $(u_1,v_1) \notin Ct(D_1\square D_2)$, since $ecc(u_1,v_2) = ecc(u_2,v_1) = 4$.  Also $(u_1,v_1)$ is not an eccentric vertex or a boundary vertex of any of the vertices in $D_1\square D_2$.  Thus $(u_1,v_1) \notin Ecc(D_1\square D_2)$ and $(u_1,v_1) \notin \partial(D_1\square D_2)$.  


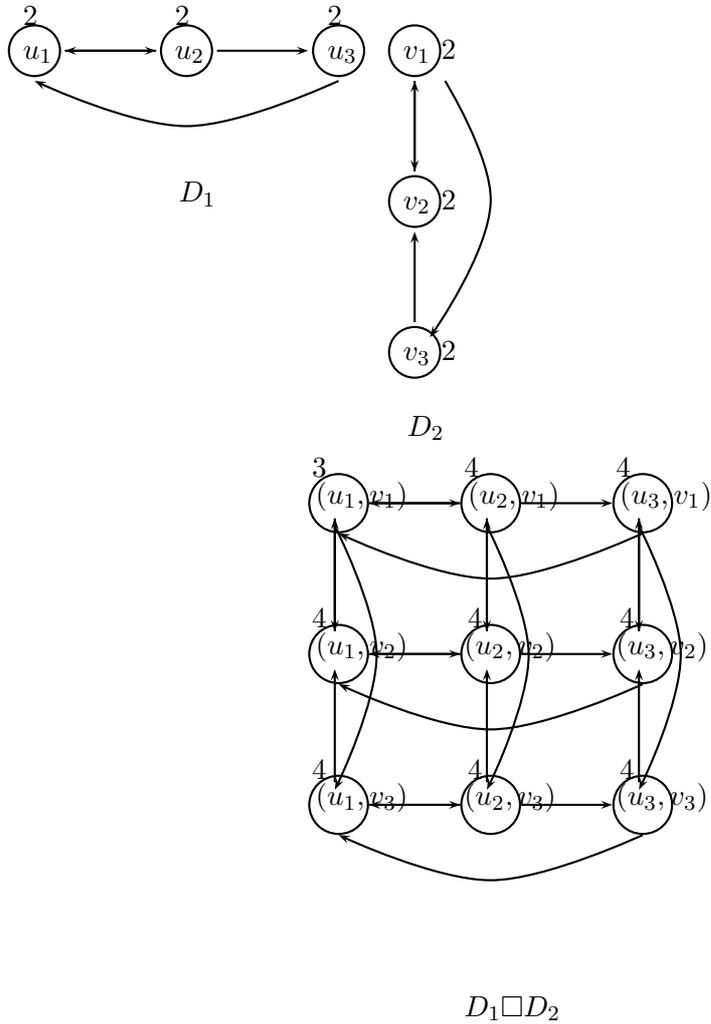
\begin{example}\label{exam3}
	\begin{figure}[H]
		\begin{pspicture}(2,16)
		\pscircle(0,14){.35}
		\put(-.15,13.9){$u_1$}
		\put(-.15,14.35){$2$}
		
		\pscircle(2,14){.35}
		\put(1.85,13.9){$u_2$}
		\put(1.85,14.35)2
		\pscircle(4,14){.35}
		\put(3.85,13.9){$u_3$}
		\put(3.85,14.35)2
		
		\psline{->}(.4,14)(1.6,14)
		\psline{->}(2.4,14)(3.6,14)
		\psline{->}(1.6,14)(.4,14)
		
		\pscurve{->}(4,13.6)(2,13)(0,13.6)
		\put(1.9,12){$D_1$}

	\pscircle(5,14){.35}
	\put(4.85,13.9){$v_1$}
	\put(5.35,13.9)2
	\pscircle(5,12){.35}
	\put(4.85,11.9){$v_2$}
	\put(5.35,11.9)2
	\pscircle(5,10){.35}
	\put(4.85,9.9){$v_3$}
	\put(5.35,9.9)2
	\psline{->}(5,12.4)(5,13.6)
		\psline{->}(5,13.6)(5,12.4)
	\psline{->}(5,10.4)(5,11.6)
	
	\pscurve{->}(5.4,13.6)(6,12)(5.2,10.2)
	\put(4.9,8.9){$D_2$}

	\pscircle(4,8){.4}
	\put(3.7,8.){$(u_1,v_1)$}
	\put(3.65,8.35)3
	\pscircle(6,8){.4}
	\put(5.7,8){$(u_2,v_1)$}
	\put(5.65,8.35)4
	\pscircle(8,8){.4}
	\put(7.7,8){$(u_3,v_1)$}
	\put(7.65,8.35)4 
	\psline{->}(4.4,8)(5.6,8)
	\psline{->}(5.6,8)(4.4,8)
	\psline{->}(6.4,8)(7.6,8)
	
	\pscurve{->}(8,7.6)(6,7)(4,7.6)
	\pscircle(4,6){.4}
	\put(3.7,6){$(u_1,v_2)$}
	\put(3.65,6.35)4
	\pscircle(6,6){.4}
	\put(5.65,6){$(u_2,v_2)$}
	\put(5.7,6.35)4
	\pscircle(8,6){.4}
	\put(7.65,6.){$(u_3,v_2)$}
	\put(7.7,6.35)4 
	\psline{->}(4.4,6)(5.6,6)
	\psline{->}(5.6,6)(4.4,6)
	\psline{->}(6.4,6)(7.6,6)
	
	\pscurve{->}(8,5.6)(6,5)(4,5.6)
	\pscircle(4,4){.4}
	\put(3.7,4){$(u_1,v_3)$}
	\put(3.65,4.35)4
	\pscircle(6,4){.4}
	\put(5.65,4){$(u_2,v_3)$}
	\put(5.7,4.35)4
	\pscircle(8,4){.4}
	\put(7.65,4){$(u_3,v_3)$}
	\put(7.7,4.35)4 
	\psline{->}(4.4,4)(5.6,4)
	\psline{->}(5.6,4)(4.4,4)
	\psline{->}(6.4,4)(7.6,4)
	
	\pscurve{->}(8,3.6)(6,3)(4,3.6)
	\put(5.65,1.2){$D_1 \square D_2$}

	\psline{->}(3.95,6.3)(3.95,7.8)
	
	\psline{->}(3.95,7.8)(3.95,6.3)
	
	\psline{->}(3.95,4.3)(3.95,5.8)
	\pscurve{->}(3.95,7.7)(4.5,6)(3.95,4.2)
	
	\psline{->}(7.95,6.3)(7.95,7.8)
	\psline{->}(7.95,7.8)(7.95,6.3)
	
	\psline{->}(7.95,4.3)(7.95,5.8)
	\pscurve{->}(7.95,7.7)(8.5,6)(7.95,4.2)
	
	\psline{->}(5.95,6.3)(5.95,7.8)
	\psline{->}(5.95,7.8)(5.95,6.3)
	
	\psline{->}(5.95,4.3)(5.95,5.8)
	\pscurve{->}(5.95,7.7)(6.5,6)(5.95,4.2)
	
\end{pspicture}
\end{figure}
\end{example}
Also from example ~\ref{exam3}, we can see that  in general
\begin{enumerate}
\item
$Ecc(D_1\square D_2)\nsubseteq  Ecc(D_1)\times Ecc(D_2) $
\item
$\partial(D_1\square D_2)\nsubseteq \partial(D_1)\times \partial(D_2) $
 \end{enumerate}
 Here $Ecc(D_1) = \partial(D_1) = \{u_1,u_3\}, Ecc(D_2) = \partial(D_2) = \{v_1,v_2,v_3\}$ .  But we can see that $(u_2,v_1) $ is an eccentric vertex of $(u_3,v_3) $ and hence a boundary vertex of $(u_3,v_3) $. Thus $u_2 \notin Ecc(D_1) = \partial(D_1)$ but
 \begin{enumerate}
 \item	
 $(u_2,v_1) \in  Ecc(D_1\square D_2)$ and
 \item
  $(u_2,v_1) \in  \partial(D_1\square D_2)$
  \end{enumerate} 

\begin{example}\label{exam4}
	\begin{figure}[H]
		\begin{pspicture}(2,11)
		\pscircle(0,10){.35}
		\put(-.15,9.9){$u_1$}
		\put(-.15,10.35){$2$}
		
		\pscircle(2,10){.35}
		\put(1.85,9.9){$u_2$}
		\put(1.85,10.35)1
		\pscircle(4,10){.35}
		\put(3.85,9.9){$u_3$}
		\put(3.85,10.35)2
		
		\psline{->}(.4,10)(1.6,10)
		\psline{->}(2.4,10)(3.6,10)
		\psline{->}(1.6,10)(.4,10)
		\psline{->}(3.6,10)(2.4,10)
		\pscurve{->}(4,9.6)(2,9)(0,9.6)
		\put(1.9,8){$D_1$}
		
		\pscircle(5,10){.35}
		\put(4.85,9.9){$v_1$}
		\put(5.35,9.9)2
		\pscircle(5,8){.35}
		\put(4.85,7.9){$v_2$}
		\put(5.35,7.9)2
		\pscircle(5,6){.35}
		\put(4.85,5.9){$v_3$}
		\put(5.35,5.9)2
		\psline{->}(5,8.4)(5,9.6)
		\psline{->}(5,6.4)(5,7.6)
		\psline{->}(5,7.6)(5,6.4)
		\pscurve{->}(5.4,9.6)(6,8)(5.2,6.2)
		\put(4.9,4.9){$D_2$}
		
		\pscircle(2,4){.4}
		\put(1.7,4){$(u_1,v_1)$}
		\put(1.65,4.35)4
		\pscircle(4,4){.4}
		\put(3.7,4){$(u_2,v_1)$}
		\put(3.65,4.35)3
		\pscircle(6,4){.4}
		\put(5.7,4){$(u_3,v_1)$}
		\put(5.65,4.35)3 
		\psline{->}(2.4,4)(3.6,4)
		\psline{->}(3.6,4)(2.4,4)
		\psline{->}(4.4,4)(5.6,4)
		\psline{->}(5.6,4)(4.4,4)
		\pscurve{->}(6,3.6)(4,3)(2,3.6)
		\pscircle(2,2){.4}
		\put(1.7,2){$(u_1,v_2)$}
		\put(1.65,2.35)3
		\pscircle(4,2){.4}
		\put(3.65,2){$(u_2,v_2)$}
		\put(3.7,2.35)3
		\pscircle(6,2){.4}
		\put(5.65,2){$(u_3,v_2)$}
		\put(5.7,2.35)4 
		\psline{->}(2.4,2)(3.6,2)
		\psline{->}(3.6,2)(2.4,2)
		\psline{->}(4.4,2)(5.6,2)
		\psline{->}(5.6,2)(4.4,2)
		\pscurve{->}(6,1.6)(4,1)(2,1.6)
	\pscircle(2,0){.4}
	\put(1.7,0){$(u_1,v_3)$}
	\put(1.65,0.35)4
	\pscircle(4,0){.4}
	\put(3.65,0){$(u_2,v_3)$}
	\put(3.7,0.35)3
	\pscircle(6,0){.4}
	\put(5.65,0){$(u_3,v_3)$}
	\put(5.7,0.35)3 
	\psline{->}(2.4,0)(3.6,0)
	\psline{->}(3.6,0)(2.4,0)
	\psline{->}(4.4,0)(5.6,0)
	\psline{->}(5.6,0)(4.4,0)
	\pscurve{->}(6,-.4)(4,-1)(2,-.4)
		
	\put(3.65,-1.5){$D_1 \square D_2$}

		\psline{->}(1.95,2.3)(1.95,3.8)
		\psline{->}(1.95,1.8)(1.95,.3)
		\psline{->}(1.95,0.3)(1.95,1.8)
		\pscurve{->}(1.95,3.7)(2.5,2)(1.95,0.2)
		
	\psline{->}(3.95,2.3)(3.95,3.8)
	\psline{->}(3.95,1.8)(3.95,0.3)
	\psline{->}(3.95,0.3)(3.95,1.8)
	\pscurve{->}(3.95,3.7)(4.5,2)(3.95,0.2)
		
		\psline{->}(5.95,2.3)(5.95,3.8)
		\psline{->}(5.95,1.8)(5.95,0.3)
		\psline{->}(5.95,0.3)(5.95,1.8)
		\pscurve{->}(5.95,3.7)(6.5,2)(5.95,0.2)
		
		\end{pspicture}
	\end{figure}
\end{example}
\vspace{1.2cm}
\begin{proposition}
Let $ D_1$ and $D_2$ be two strongly connected digraphs such that atleast one of $ D_1$ and $D_2$ have the two-sided eccentricity property.Then 
\begin{enumerate}
\item
$Per(D_1\square D_2) = Per(D_1)\times Per(D_2) $
\item
$Ct(D_1\square D_2) = Ct(D_1)\times Ct(D_2) $
\end{enumerate}
\end{proposition}
\begin{proof}
\begin{enumerate}
\item
We have $Per(D_1\square D_2)\subseteq Per(D_1)\times Per(D_2) $ in every case. So it remains to prove that $Per(D_1)\times Per(D_2)\subseteq Per(D_1\square D_2)$.   Let $u_i \in Per(D_1) $ and $v_r \in Per(D_2)$.   Hence $ecc(u_i) > ecc(u_j)$, for all $u_j \in V(D_1)$ and $ecc(v_r) > ecc(v_s)$, for all $v_s \in V(D_2)$ which gives $ecc(u_i)+ecc(v_r) > ecc(u_j)+ecc(v_s)$, for all $u_j \in V(D_1)$ and  for all $v_s \in V(D_2)$.  Since atleast one of $ D_1$ and $D_2$ have the two-sided eccentricity property,

$ecc_{D_1\Box
	D_2}(u_i,v_r)=ecc_{D_1}(u_i)+ecc_{D_2}(v_r)$ ,
for all $((u_i,v_r) \in V(D_1 \Box D_2)$.\\
So we get
$ecc(u_i,v_r)>ecc(u_j,v_s)$, for all $(u_j,v_s) \in V(D_1 \square D_2)$ so that $(u_i,v_r) \in  Per(D_1\square D_2)$. 

\item
 We have already shown that $Ct(D_1\square D_2)\subseteq Ct(D_1)\times Ct(D_2) $.  Conversely suppose that $u_i \in Ct(D_1)$ and $v_r \in Ct(D_2)$.  If possible, let $(u_i,v_r) \notin Ct(D_1\square D_2)$.  Then there is a vertex $(u_j,v_s) \in N(u_i,v_r) $ such that $ecc(u_j,v_s) > ecc(u_i,v_r)$.  Since  $(u_j,v_s) \in N(u_i,v_r) $ without loss of generality, assume that $u_i =u_j $ and $v_s \in N(v_r) $. So we get $ecc(u_i)+ecc(v_s) > ecc(u_i)+ecc(v_r)$ which gives $ecc(v_s) > ecc(v_r)$ which is a contradiction.  Hence $ Ct(D_1)\times Ct(D_2)\subseteq Ct(D_1\square D_2)$. 
\end{enumerate}
\end{proof}
\begin{proposition}
Let $ D_1$ and $D_2$ be two strongly connected digraphs. Let $u_i \in V(D_1) , v_r \in V(D_2)$.  Suppose that both $ D_1$ and $D_2$ satisfy the two-sided eccentricity property. \\ That is both the conditions ~\ref{suf} and ~\ref{suf1}
 are satisfied.\\  
 Then  $ecc(u_i,v_r)=ecc(u_i)+ecc(v_r)=d(u_i,u_j)+d(v_r,v_q)=d((u_i,v_r) , (u_j,v_q)) = d(u_k,u_i)+d(v_s,v_r) =d((u_k,v_s) , (u_i,v_r))$
\end{proposition}
\begin{proof}
We have already shown that $ecc(u_i,v_r)=ecc(u_i)+ecc(v_r)$ if atleast one of the above conditions is satisfied.\\ 
$ecc(u_i,v_r)=ecc(u_i)+ecc(v_r)=\overrightarrow{d}(u_i,u_j)+\overrightarrow{d}(v_r,v_q)$. \\ Since $ecc(u_i)=\overrightarrow{d}(u_i,u_j)$, we get $d(u_i,u_j)=max\{\overrightarrow{d}(u_i,u_j),\overrightarrow{d}(u_j,u_i)\} =\overrightarrow{d}(u_i,u_j)$ and since $ecc(v_r)=\overrightarrow{d}(v_r,v_q)$ we get $d(v_r,v_q)=max\{\overrightarrow{d}(v_r,v_q),\overrightarrow{d}(v_q,v_r)\} =\overrightarrow{d}(v_r,v_q)$.  \\Hence $ecc(u_i,v_r)=ecc(u_i)+ecc(v_r)=d(u_i,u_j)+d(v_r,v_q)$.\\ Similarly we can show that $ecc(u_i,v_r)=ecc(u_i)+ecc(v_r) = d(u_k,u_i)+d(v_s,v_r)$. \\  Also, $d((u_i,v_r) , (u_j,v_q))=max\{\overrightarrow{d}(u_i,u_j)+\overrightarrow{d}(v_r,u_q),\overrightarrow{d}(u_j,u_i)+\overrightarrow{d}(v_q,u_r)\} =\overrightarrow{d}(u_i,u_j)+\overrightarrow{d}(v_r,u_q) = d(u_i,u_j)+d(v_r,v_q)$ and similarly $d(u_k,u_i)+d(v_s,v_r) =d((u_k,v_s) , (u_i,v_r))$.
\end{proof}

\begin{corollary}
Let $D_1$ and $D_2$ be two strong digraphs having the two-sided eccentricity property.  Then in addition to periphery and contour, $Ecc(D_1\square D_2)= Ecc(D_1)\times Ecc(D_2)$. 
\end{corollary}

\begin{proof}
Let $(u_i,v_r) \in Ecc(D_1\square D_2)$. \\ So there exists a vertex $(u_j,v_q)$ such that $ecc(u_j,v_q)
=d((u_j,v_q),(u_i,v_r)) $. \\ Hence $ecc(u_j)+ecc(v_q)=d(u_j,u_i)+d(v_q,v_r)$. \\Then necessarily 
$ecc(u_j)=d(u_j,u_i)$ and $ecc(v_q)=d(v_q,v_r)$ which gives $u_i \in Ecc(D_1))$ and $v_r \in Ecc(D_2)$. 

Conversely if $u_i \in Ecc(D_1))$ and $v_r \in Ecc(D_2)$
then there are vertices $u_j$ and $v_q$ respectively such that $ecc(u_j)=d(u_j,u_i) $ and $ecc(v_q)=d(v_q,v_r) $. \\ Hence we get\\ $ecc(u_j,v_q)=ecc(u_j)+ecc(v_q)=d(u_j,u_i)+d(v_q,v_r)=d((u_j,v_q),(u_i,v_r)) $ \\which gives $(u_i,v_r) \in Ecc(D_1\square D_2)$.
\end{proof}
\begin{proposition}
$D_1\square D_2$ have the two-sided eccentricity property if and only if both $D_1$ and  $D_2$  have the two-sided eccentricity property. 
\end{proposition}
\begin{proof}
	Suppose that  both $D_1$ and  $D_2$  have the two-sided eccentricity property.  For every $u_i \in V(D_1) , v_r \in V(D_2)$ there exist vertices $u_j,u_k$ in $D_1$ ($u_j$ may be equal to $u_k$) such that
	\begin{eqnarray}\label{necsuf}
	ecc(u_i)=\overrightarrow{d}(u_i,u_j)=\overrightarrow{d}(u_k,u_i) 
	\end{eqnarray} and there exist vertices $v_q,v_s$ in $V(D_2)$ ($v_q$ may be equal to
	$v_s$) such that
	\begin{eqnarray}\label{necsuf1}
	ecc(v_r)=\overrightarrow{d}(v_r,v_q)=\overrightarrow{d}(v_s,v_r) 
	\end{eqnarray}
	We have shown that $ecc(u_i,v_r)=\overrightarrow{d}(u_i,u_j)+\overrightarrow{d}(v_r,v_q)=\overrightarrow{d}((u_i,v_r) , (u_j,v_q)) =\overrightarrow{d}(u_k,u_i)+d(v_s,v_r) =\overrightarrow{d}((u_k,v_s) , (u_i,v_r))$.  Thus $D_1\square D_2$  have the two-sided eccentricity property. 
	
	Conversely we have to show that if $D_1\square D_2$  have the two-sided eccentricity property then both $D_1$ and $D_2$ have the two-sided eccentricity property. \\ For this we show that if any one of $D_1$ and $D_2$ does not have the two-sided eccentricity property, then 
	$D_1\square D_2$ does not have the two-sided eccentricity property.\\  Without loss of generality, suppose that $D_1$ does not have the two-sided eccentricity property. \\ Hence there exist atleast one vertex, say $u_i \in V(D_1)$ such that $ecc(u_i)=\overrightarrow{d}(u_i,u_j)>\overrightarrow{d}(u_k,u_i) $ for every $u_k \in V(D_1)$. \\ Let $v_r$ be any arbitrary vertex in $D_2$ and suppose that there exist vertices $v_q,v_s$ in $V(D_2)$ ($v_q$ may be equal to
	$v_s$) such that
$	ecc(v_r)=\overrightarrow{d}(v_r,v_q)=\overrightarrow{d}(v_s,v_r) $.  
	
	Consider $(u_i,v_r) \in V(D_1\square D_2)$.  We have $\overrightarrow{d}((u_i,v_r) , (u_j,v_q))=\overrightarrow{d}(u_i,u_j)+\overrightarrow{d}(v_r,v_q)$ and $\overrightarrow{d}((u_k,v_s) , (u_i,v_r))=\overrightarrow{d}(u_k,u_i)+d(v_s,v_r). $ 
	
	Hence $\overrightarrow{d}((u_i,v_r) , (u_j,v_q))>\overrightarrow{d}((u_k,v_s) , (u_i,v_r))$.  \\
	Then in $D_1\square D_2$, we cannot find any vertex $(u_k,v_s)$ such that 
	$ecc(u_i,v_r)=\overrightarrow{d}((u_i,v_r) , (u_j,v_q))  =\overrightarrow{d}((u_k,v_s) , (u_i,v_r))$.
	\end{proof}

\begin{proposition}
Let $D_1 $ be an undirected graph and $D_2$ be a strong digraph.  Then $d((u_i,v_r) , (u_j,v_q))=d(u_i,u_j)+d(v_r,v_q)$, for every $u_i,u_j \in V(D_1),v_r,v_q \in V(D_2)$.
\end{proposition}

\begin{proof}
\begin{align*}
d((u_i,v_r) , (u_j,v_q))&=max\{\overrightarrow{d}(u_i,u_j)+\overrightarrow{d}(v_r,u_q),\overrightarrow{d}(u_j,u_i)+\overrightarrow{d}(v_q,u_r)\} \\&=max\{\overrightarrow{d}(u_i,u_j)+\overrightarrow{d}(v_r,u_q),\overrightarrow{d}(u_i,u_j)+\overrightarrow{d}(v_q,u_r)\} \\&=d(u_i,u_j)+max\{\overrightarrow{d}(v_r,u_q),\overrightarrow{d}(v_q,u_r)\}\\&=d(u_i,u_j)+d(v_r,v_q)
\end{align*}.
\end{proof}

\begin{proposition}
Let $D_1 $ be an undirected graph and $D_2$ be a strong digraph.  Then
 \begin{enumerate}
 \item
 $\partial(D_1\square D_2)= \partial(D_1)\times \partial(D_2) $ . Also 
\item 
$Ecc(D_1\square D_2)= Ecc(D_1)\times Ecc(D_2) $
\item
$Per(D_1\square D_2) = Per(D_1)\times Per(D_2) $ 

\item
$Ct(D_1\square D_2) = Ct(D_1)\times Ct(D_2) $
\end{enumerate}
\end{proposition}
\begin{proof}

\begin{enumerate}
 \item
Let $(u_i,v_r) \in \partial(D_1\square D_2)$ and $u_i \notin \partial D_1$. Then for every $u_j \in V(D_1)$ there exists $u_k \in N(u_i)$ such that $d(u_j,u_k)>d(u_j,u_i)$.  Consider an arbitrary vertex $(u_k,v_q) \in
N(u_i,v_r)$.  Let $v_q $ be an arbitrary vertex in $V(D_2)$.  Then $d((u_j,v_q),(u_k,v_r))=d(u_j,u_k)+d(v_q,v_r)>d(u_j,u_i)+d(v_q,v_r)=d((u_j,v_q),(u_i,v_r))$ which contradicts  $(u_i,v_r) \in \partial(D_1\square D_2)$. Hence $u_i \in \partial (D_1)$. Similarly we can prove $v_r \in \partial (D_2)$.

Conversely, let $u_i \in \partial (D_1)$ and $v_r \in \partial (D_2)$.  Thus there exists a vertex $u_j \in V(D_1)$ such that for every $u_k \in N(u_i)$, $d(u_j,u_i) \geq d(u_j,u_k)$.  Also there exists a vertex $v_q \in V(D_2)$ such that for every $v_s \in N(v_r)$, $d(v_q,v_r) \geq d(v_q,v_s)$.  Consider an arbitrary vertex $(u_k,v_s) \in N(u_i,v_r)$.  Without loss of generality, assume that $u_k $ is adjacent to $u_i$ in $D_1$ and $v_r=v_s$ in $D_2)$. Then we get  $d((u_j,v_q),(u_k,v_s))=d(u_j,u_k)+d(v_q,v_s) \leq d(u_j,u_i)+d(v_q,v_r)= d((u_j,v_q),(u_i,v_r)) $ which gives $(u_i,v_r) \in \partial(D_1\square D_2)$. 
\item 
Since $D_1$ is an undirected graph, $ecc(u_i,v_r)=ecc(u_i)+ecc(v_r)$ for every $(u_i,v_r) \in V(D_1\square D_2)$. Let $(u_i,v_r) \in Ecc(D_1\square D_2)$. So there exists a vertex $(u_j,v_q)$ such that $ecc(u_j,v_q)
=d((u_j,v_q),(u_i,v_r)) $.  Hence $ecc(u_j)+ecc(v_q)=d(u_j,u_i)+d(v_q,v_r)$ since $D_1$ is an undirected graph.  Then necessarily 
$ecc(u_j)=d(u_j,u_i)$ and $ecc(v_q)=d(v_q,v_r)$ which gives $u_i \in Ecc(D_1))$ and $v_r \in Ecc(D_2)$.

Conversely if $u_i \in Ecc(D_1))$ and $v_r \in Ecc(D_2)$ then there are vertices $u_j$ and $v_q$ respectively such that $ecc(u_j)=d(u_j,u_i) $ and $ecc(v_q)=d(v_q,v_r) $. Hence we get $ecc(u_j,v_q)=ecc(u_j)+ecc(v_q)=d(u_j,u_i)+d(v_q,v_r)=d((u_j,v_q),(u_i,v_r)) $ which gives $(u_i,v_r) \in Ecc(D_1\square D_2)$.
\end{enumerate}

Cases 3 and 4 holds since $D_1$ have the two-sided eccentricity property. 
\end{proof}

\section{Some More Results}
We have shown that if one of the digraphs have the two-sided eccentricity property then 
\begin{enumerate}
	\item
	$Per(D_1\square D_2) = Per(D_1)\times Per(D_2) $ 
	
	\item
	$Ct(D_1\square D_2) = Ct(D_1)\times Ct(D_2) $
	
\end{enumerate}
If both the digraphs have the two-sided eccentricity property then we also get $Ecc(D_1\square D_2) = Ecc(D_1)\times Ecc(D_2) $.  Also, if one of the digraphs is an undirected graph then in addition, $\partial(D_1\square D_2) = \partial(D_1)\times \partial(D_2) $.  Now we are at a stage to extend these results to the cartesian product of $n$ directed graphs .  Let $D_1(V_1 , E_1) , D_2(V_2 , E_2) , \ldots , D_n(V_n , E_n)$ be n directed graphs.  Let $x_1 \in  V_1 ,x_2 \in  V_2 , \ldots ,x_n \in  V_n $.  Then $(x_1,x_2, \ldots ,x_n) \in V(D_1 \square D_2 \ldots \square D_n)$.

If all except one of $D_1 , D_2, \ldots D_n$ have the two-sided eccentricity property, then  $ecc(x_1,x_2, \ldots ,x_n) = ecc(x_1)+ ecc(x_2)+  \ldots + ecc(x_n)$, since cartesian product is associative and commutative.  So as in the case of two directed graphs  $D_1$ and $D_2$, we get 
$ Per(D_1 \square D_2 \ldots \square D_n) = Per(D_1) \times Per(D_2) \ldots \times Per(D_n) $ and $Ct(D_1 \square D_2 \ldots \square D_n) = Ct(D_1) \times Ct(D_2) \ldots \times Ct(D_n) $.  If all of $D_1 , D_2, \ldots D_n$ have the two-sided eccentricity property, then we also get $Ecc(D_1 \square D_2 \ldots \square D_n) = Ecc(D_1) \times Ecc(D_2) \ldots \times Ecc(D_n) $.  
An interesting consequence is that if $D_1 , D_2, \ldots D_n$ are either cycles or undirected graphs, then all the above hold.  
Another interesting result is that if the digraph $D=D_1 \square D_2 \ldots \square D_n$ is the cartesian product of $n$ cycles, then in addition to above, we get 
$\partial(D_1 \square D_2 \ldots \square D_n) = \partial(D_1) \times \partial(D_2) \ldots \times \partial(D_n) $.  This is because in the case of a cycle $\overrightarrow{C}$, $\partial(\overrightarrow{C})= Ecc(\overrightarrow{C})=Ct(\overrightarrow{C})=Per(\overrightarrow{C})$.  
If all except one of the factors in the prime factor decomposition turn out to be undirected graphs, then also $\partial(D_1 \square D_2 \ldots \square D_n) = \partial(D_1) \times \partial(D_2) \ldots \times \partial(D_n) $. 
Even though we discussed about the four boundary type sets, we can see that the periphery and contour sets are more significant as they can be considered as global concepts regarding the strong digraph under consideration, whereas the other two are local concepts.
\section{Conclusion} 

The significance of the above results lies in applying these results together with prime factor decomposition of digraphs.  Thus given a large strongly connected digraph, the informations regarding the boundary type sets can be obtained  much more easily.  The drawback is that it is applicable only when atmost one of them do not have the two sided eccentricity property in which case the usual methods have to be applied. 

\bibliographystyle{amsplain}
\bibliography{directed}
\end{document}